\documentclass[11pt]{article}
\usepackage[utf8]{inputenc}

\usepackage{fullpage}
\usepackage{color}
\usepackage{amsmath}
\usepackage{amssymb}
\usepackage{amsfonts}
\usepackage{amsthm}
\usepackage[all]{xy}
\usepackage{float}
\usepackage[small,compact]{titlesec}

\usepackage{boxedminipage}

\usepackage{listings}

\usepackage{ifpdf}

\theoremstyle{definition}

\theoremstyle{plain}
\newtheorem{thm}{Theorem}[section]
\newtheorem{claim}{Claim}
\newtheorem{obs}{Observation}
\newtheorem{cor}{Corollary}
\newtheorem{lemma}{Lemma}

\theoremstyle{remark}
\newtheorem{remark}{Remark}
\usepackage[ruled,linesnumbered]{algorithm2e}

\newcommand{\comment}[1]{}

\newcommand*\samethanks[1][\value{footnote}]{\footnotemark[#1]}
 
\begin{document}

\title{Efficient Sum-Based Hierarchical Smoothing Under $\ell_1$-Norm}
\author{
Siavosh Benabbas\thanks{
Department of Computer Science, University of Toronto}
\\{\tt siavosh@cs.toronto.edu}
\and 
Hyun Chul Lee\thanks{
Thoora Inc., Toronto, ON, Canada}
\\{\tt chul.lee@thoora.com}
\and 
Joel Oren\samethanks[1] \thanks{This research was supported by the
  MITACS Accelerate program, Thoora Inc., and The University of Toronto, Department of Computer Science.}
\\{\tt oren@cs.toronto.edu}
\and 
Yuli Ye\samethanks[1] \samethanks[3]
\\{\tt y3ye@cs.toronto.edu}
}

\maketitle
\begin{abstract}
We introduce a new regression problem which we call the
 \textit{Sum-Based Hierarchical Smoothing} problem. Given a directed
 acyclic graph and a non-negative value, called \emph{target value}, for each vertex in the graph, 
 we wish to find non-negative values for the vertices satisfying a certain constraint while minimizing
 the distance of these assigned values and the target values in the $\ell_p$-norm.
 The constraint is that the value assigned to each vertex should be no less than the sum of the values assigned to its children.
 We motivate this problem with applications in information retrieval and web mining.
 While our problem can be solved in polynomial time using linear programming,
 given the input size in these applications such a solution is too slow.

% This simple constrained smoothing problem gives way to a much richer and theoretically more challenging set of problems than 
% other similar smoothing problems, where the most notable of them is the 
%problem. % Joel: not sure if this is a good enough reference.
%We illustrate the motivation for our problem through settings
%in information retrieval and web mining. 

 We mainly study the $\ell_1$-norm case restricting the underlying graphs to rooted trees.
 For this case we provide an efficient algorithm, running in $O(n^2)$ time. 
 While the algorithm is purely combinatorial, its proof of correctness is an elegant
 use of linear programming duality. We also present a number of other positive and negatives
 results for different norms and certain other special cases.

 We believe that our approach may be applicable to similar problems, where
 comparable hierarchical constraints are involved, e.g. considering the
 average of the values assigned to the children of each vertex. While similar in flavour to other smoothing 
 problems like Isotonic Regression (see for example [Angelov et al. SODA'06]), our problem is arguably
 richer and theoretically more challenging.
\end{abstract}
\setcounter{page}{0}
\thispagestyle{empty}
\newpage
%
%\category{J.4}{Social and Behavioral Sciences}{Economics}
%\category{F.2.2}{Analysis of Algorithms and Problem Complexity}{Nonnumerical Algorithms and Problems}
%\terms{games, economics, algorithms}
%\keywords{influence spread} % NOT required for Proceedings

\section{Introduction}
The prevalence of popular web services like Amazon, Google, Netflix, and StumbleUpon has given rise to 
many interesting large-scale problems related to classification, recommendation, ranking, and collaborative filtering. In several recent studies (e.g. \cite{kfb, pg, ckp}), researchers have incorporated the underlying class hierarchies of the data-sets into the setting of recommendation systems.
Moreover, Koren et al. \cite{dkk11} recently demonstrated an application of hierarchical classifications of topics,
i.e. \emph{taxonomies}, in Collaborative Filtering settings, in particular, music recommendation. \comment{In this setting, linking the items to a taxonomy alleviates some of the difficulties that arise from the high sparsity of the data.}
In these application scenarios, the taxonomies\comment{objects of the given dataset} are abstracted as trees.
Associated with the vertices are scalar target values, %, that correspond to the objects in the given dataset. 
typically inferred through the use of various machine learning or information retrieval methods.
For instance, given a hierarchy of topics and a search query, the target values
could be the relevance measures of the topics to the search query.

When a taxonomy is used, one would usually like to enforce particular constraints on the value
assigned to the vertices to properly represent the hierarchical relationship among them. 
Typically, the relevant machine learning approaches are ill-equipped to handle these 
requirements. Often, these constraints state that the value of each vertex should be at least some function
of the value of its direct children in the taxonomy (e.g. \cite{pg, ckp}).

%these constraints are one-to-many relations, in which the value of each vertex
% is compared to a combination of a set of vertex values. %Siavosh: This last sentence adds very little information. Maybe we can make it less vague by just saying children or the vertices it has edges from instead of a set of vertex values. Also, what is one-to-many relation mean here?

Going back to the previous example of topics and search query, imagine that the taxonomy contains the topics
\emph{sports}, \emph{baseball}, \emph{football}, and \emph{basketball} with the first topic being the parent of
the other three and that the search query is ``ESPN''. One would like to find the relevance of this query to every
topic in the taxonomy. A reasonable requirement of these relevance values would be that 
%
% One can imagine that a regression algorithm has found
%relevance scores of each of the four topics to ``ESPN'' and these numbers are available. 
%It would be that
 the relevance of ``ESPN'' to \emph{sports} would be no less than the sum of its relevance to
\emph{baseball}, \emph{football}, and \emph{basketball}. One way to solve this problem would be to directly
impose such a constraint on the learning algorithm that infers the relevance values using regularization; i.e. adding
an additional term in the objective function of that algorithm penalizing any violation of the constraint.
However, this approach has two problems. First, it ``softens'' our requirements; i.e. it allows for possible
violations, to some limited extend. Moreover, it can dramatically deteriorate the running time of the process
of learning or restrict our choice of the learning algorithm.

Instead, we take the following, widely used, two-step approach. %%Siavosh: widely used? Is it really widely used and should we perhaps quantifly that with ``in other settings'' somehow? I don't personally feel that strongly about this.
Given a search query $s$, we first infer each of the relevance scores of each of the 
topics, disregarding the hierarchy constraints. Then, we \emph{smoothen}
the inferred relevance scores by 
modifying them so as to uphold the above sum constraint. 
We would want the change of the relevance scores in the second step to be as small as possible.
 As the relevance scores are scalar values, we can represent both the original and 
final relevance scores as two vectors with non-negative values, and 
measure their difference in a suitable norm (e.g. the $\ell_1, \ell_2$ or the $\ell_\infty$ norms). 
The subject of this paper is how to perform the second step.

%%Siavosh: Here is a good place to mention that this would have applications beyond the simple example, e.g. a forthcoming paper on recommendation systems.

We formulate this problem which we call the \textit{Sum-Based Hierarchical Smoothing} problem (SBHSP) as follows.
Given a rooted tree (or in general a directed acyclic graph) $G = (V, E)$ and a vector of original vertex values
(called \emph{target values}) ${\bf a} = (a_{v_1}, a_{v_2}, \ldots, a_{v_n})$ the objective is to find a vector of new
vertex values (called \emph{assigned values}) ${\bf x} = (x_{v_1}, x_{v_2}, \ldots, x_{v_n})$ with the following properties.
(i) for any node $w$ with incoming edges $(u_1, w), \ldots, (u_k, w)$ we have $x_{u_1} + \cdots + x_{u_k} \leq x_{w}$. (ii)
$||{\bf a} - {\bf x}||_p$ is minimized. Different values of $p$ result in different variants of the problem. We mainly study the problem
for $p = 1$ and $p = \infty$ but the case of $p = 2$ is also interesting. It is not hard to see that for $p=1$ the problem can be
solved in polynomial time using linear programming (see inequalties \eqref{bound1}-\eqref{bound3p}) and for $p>1$ it can be solved by using a
suitable separation oracle and the Ellipsoid method. However given the typical size of taxonomies these solutions are too slow.
We note that this problem seems to be more complex than other previously considered similar problems as the assigned value of each vertex
affects the possible values for any vertex it shares a parent with. In particular, to the best of our knowledge techniques used for
similar problems are ineffective for it.

\paragraph{Contributions:} Our main contribution is a purely combinatorial algorithm when the input is a rooted tree and $p=1$
(i.e.~the $\ell_1$ norm) that runs in time $O(n^2)$. We note that the
$\ell_1$ norm was previously used as a good measure of difference in similar regression problems (e.g. see \cite{ahkw}). As many
hierarchical structures in practice are trees, our algorithm can be used in many practical applications.
Our second contribution is a linear time algorithm for the case $p=\infty$ which works for any directed acyclic graph. 
We also show an efficient FPTAS for optimizing the $\ell_1$ norm for another class of DAGs (directed bilayer graphs.)
%Furthermore, we give a simple reduction in which there is an efficient FPTAS in the case where the underlying graphs are directed bilayer graphs.
Finally, we show that if one adds the extra condition that the assigned values should be \emph{integral} the problem is hard to approximate
(to within a polylogarithmic factor) for any $\ell_p$ norm for $1 \leq p < \infty$. Interestingly, given that our algorithm for the $\ell_1$ norm
on trees always outputs an integral solution this last result suggests that new ideas are needed to extend it to general DAGs.
%result on approximation for the general case when integral solutions are required, even when the underlying graph is a directed bipartite graph.
%closest (to the given vector of vertex values) vector of new vertex values, 
%while maintaining the hierarchical sum constraint. We propose an optimal, combinatorial method for this problem.
% We observe that there is a general LP formulation of sum-based hierarchical smoothing 
%when the underlying graph is a directed acyclic graph (DAG), and the distance measure to the input vertex values is the $\ell_1$-norm.

%As methods for solving large linear programs do not scale well, naively solving the corresponding LP can be
%impractical. One of the main contributions of our paper is to show that there is an efficient and succinct $O(n^2)$ combinatorial algorithm. 

Our algorithm for the $\ell_1$ case has a rather simple structure. We assign values to the vertices of the tree in a bottom-up manner. For each
vertex we first assign a valid (but possibly suboptimal) value and then use paths going down from that vertex to ``push the excess'' down the tree
and improve the objective value.
While the algorithm is purely combinatorial, its proof of correctness is an elegent use of linear programming duality. In particular, we use the
complementary slackness condition to show that if the algorithm can no longer push the excess of a node down the tree the values assigned to its
subtree most be optimal.

\paragraph{Organization:} We present the relevant previous work in Section~\ref{sec:previous}. In Section~\ref{sec:prelim}, we present
a precise definition of the problem and some preliminaries. We present our first algorithm which is for the case of trees and $\ell_1$ norm in
Section~\ref{sec:algabs} and prove its correctness. In Section~\ref{sec:algfinal} we show how this algorithm can be optimized to run in the
promissed $O(n^2)$ time. We conclude and propose several open problems in Section~\ref{sec:conclusion}. 
We extend the algorithm to the case of \emph{weighted} $\ell_1$ norm in Appendix~\ref{sec:weighted}. We present our algorithm
for the case of $\ell_\infty$ in Appendix~\ref{sec:ellinfty}. We leave our hardness of approximation result to Appendix~\ref{sec:hardness}
and our results for the case of bilayer graphs to Appendix~\ref{sec:bilayer}.
%In the appendix, we present our additional result regarding the case where dealing with bilayer graphs, when all the edges are directed from one layer, to the other.

%%% Local Variables: 
%%% mode: latex
%%% TeX-master: "Regression"
%%% End: 

\section{Previous Work}\label{sec:previous}
The main motivation of the current paper is the application of taxonomies in regression. A recent example, studied by 
Koren et al.~\cite{dkk11}, is the application of topic hierarchies in the context of collaborative filtering. %They study the Yahoo! music dataset.
%, which comprises of over six hundred thousand musical items (e.g. songs, albums, etc.),  and two hundred fifty million user ratings, which were gathered over the course of a decade. 
They provide a method of linking the data-set to a four level taxonomy, which helps them circumvent difficulties related to the
 size of the data-set. % large number of items.

Regression and smoothing problems have been studied extensively in recent years. Perhaps the most relevant problem to our setting is the {\em Isotonic regression} problem and its variants.
% is the most relevant problem to our {\em sum-based hierarchical smoothing} problem. 
There one wishes to find a  closest fit to a given vector subject to a set 
of monotonicity constraints. More precisely, let $\mathbf{a}=\langle a_1, \dots, a_n \rangle$ be $n$ target values% with corresponding weights
%$\mathbf{w}= \langle w_1, \dots, w_n \rangle$
, and let $E$ be a set of $m$ pairwise order constraints on these variables. The
{\em Isotonic regression} problem is to find values $\mathbf{x}=\langle x_1, \dots, x_n \rangle$ such that $x_i\ge x_j$ whenever
$(i,j)\in E$ for which the distance between $\mathbf{x}$ and $\mathbf{a}$ is minimized. To put things
in a language similar to ours, in isotonic regression the assigned value of each vertex should be bigger than the \emph{maximum} of the
assigned value of its children as opposed to the \emph{sum} of those values in our problem. 
%The (perhaps more complex) summation constraint in our setting naturally arises for any hierarchical data structure.

Common choices of distance functions include the weighted $\ell_1$, $\ell_2$ and $\ell_\infty$ norms. The Isotonic regression problem for such 
weighted norms have been studied extensively. % ranging from the simple path cases to the most general directed acyclic graph (DAG) cases. 
For some of the results for the $\ell_1$ and $\ell_2$ norms see \cite{s1,ahkw,best}. Stout also maintains a web site containing
some of the fastest known Isotonic regression algorithms for different settings at \cite{stoutxxisoregalg}.
%While the most widely used algorithmic approach 
%for the  $\ell_2$-norm case is the {\em Pooled Adjacent Violators} (PAV) algorithm, which runs in linear time, Best and Chakravarti \cite{best} have approached the problem 
%as an active set identification problem proposing a linear time algorithm that has the same complexity as that of the PAV algorithm.

%Note that the 
The Isotonic regression problem belongs to a more general class of problems known as {\em order restricted statistical inference}. 
Order restricted statistical inference was first studied by Barlow et al~\cite{bbbb}. The Isotonic regression problem 
became popular since it has many applications in testing~\cite{leuraud, mancuso}, modelling~\cite{morton, ulm}, data smoothing~\cite{friedman, pg} and other areas~\cite{rwd} related to statistical and computational data analysis. It has been shown to be an important 
post-processing smoothing tool to impose desired hard constraints on the values that a learning algorithm has produced. %For instance, Punera and 
%Ghosh used the hierarchical arrangement of topical classes to represent the intuitive relationships between scores produced by a classifier 
%and have shown through several experiments that the classification accuracy can be greatly enhanced using such hierarchical arrangement of topical 
%classes. In order to enforce the necessary hierarchical relationships, they used a regularized 
%Isotonic regression approach restricting the classifier's output for each node to be greater or equal to the maximum of the classifier's output scores of its children 
%nodes. 
%Their regularized Isotonic regression problem is solved through a dynamic programming algorithm. 
Variations of Isotonic regression have been used for other applications like template learning~\cite{ckp}, ranking~\cite{dczmbzbld, mscz}, and classification~\cite{kfb}.

\section{Preliminaries}\label{sec:prelim}
We now formally define the problem as follows. Given a tree (or DAG)
$T=(V,E)$ rooted at node $r \in V$,
and a vector $\mathbf{a} \in \mathbb{R}^{n}_{\geq 0}$ of the target values of the vertices. We wish
to find the \emph{closest} vector $\mathbf{x} \in \mathbb{R}^{n}_{\geq 0}$, in
the $\ell_p$-norm,
 so that for each node $v$, with
children $u_1,\ldots,u_k$, $x_v \geq x_{u_1} + x_{u_2} + \cdots + x_{u_k}$.

While most of the paper addresses the case of $p=1$, we also
discuss the case of 
 $p=\infty$ in Appendix~\ref{sec:ellinfty}. Note that our hardness results 
apply to \emph{all} $1 \leq p < \infty$.

For a vertex $u \in T$, we denote the set of nodes with edges to $u$ the \emph{children} of
$u$ or $C(u)$, similarly the parent of $u$ is $A(u)$ (in the case where the
underlying graph is a general DAG, $A(u)$ will be a set of nodes). Throughout the paper, we will make extensive use of various paths
in the given tree. For this purpose, we let $P_{u \rightarrow v}$ denote the (unique) path from vertex $u$ to vertex
$v$ in $T$. We denote the sub-tree rooted in vertex $v$ by $T_v$. For a given sub-tree $T_v$, we define
$\mathbf{a}|_{T_v}$ as the vector of target values corresponding to the nodes in $T_v$; we similarly define $\mathbf{x}|_{T_v}$. 

%%% Local Variables: 
%%% mode: latex
%%% TeX-master: "Regression"
%%% End: 

\section{The Algorithmic Approach for $\ell_1$}
\label{sec:algabs}
As an initial attempt, consider the following 
trivial feasible solution. 
For each leaf $\ell \in T$, set $x_\ell = a_\ell$. 
Then, for each internal node $v$ set $x_v = \max(a_v,\sum_{u \in C(v)}x_u)$, 
by traversing the tree in post-order. 
However, it is not hard to see that this approach would be arbitrarily 
sub-optimal (see Figure~\ref{fig1}.) Indeed, in some cases it is preferable 
to lower the existing $x$ values of a 
given node's children, instead of raising the node's 
$x$ value, as this might help the objective value on the nodes ancestors as well.

In order to optimize the objective function, our algorithm will proceed as follows. 
By traversing the tree $T$ in post-order, it
performs the following sequence of steps for every vertex $v$. $x_v$ is
initially set to the maximum of $a_v$ and the sum of the
$x$ values of its children, which is clearly a feasible assignment for $T_v$. 
It then improves the assignments for $T_v$ 
by sequentially decreasing the values
of some vertices that are located on some path $P$ from $v$ to some other node in
$T_v$. The adjustments are made so that the overall improvement in the objective function
equals the improvement in $|a_v - x_v|$. %Notice that as long as $x_v$ is greater than $a_v$, $x_v$ will decrease in any improvement made on it. 
We will refer to such paths as \emph{push-paths}, and the improvements made on them as \emph{push operations}. 
The algorithm is presented below as Algorithm~\ref{alg1}. 
%In the statement of the algorithm (Algorithm~\ref{alg1}), 
The procedure $Push-Path(\mathbf{x}, P, \epsilon)$ checks what is the improvement on the objective function value
if we reduce the $x$ value of all vertices in the path $P$ by $\epsilon$. %, and adjust its successors in order to
%maintain feasibility. 
This path will always start at the current vertex $v$.

For now we do not discuss how to find the push path or the exact value that we push down that path.
% Notice that nowhere in the algorithm do we specify the way by which we find the exact value that we push
%down from the current node or how to find the push path. 
This abstraction was made deliberately, so as to to separate the correctness of the algorithm from its performance. 
In fact, we later show that the individual paths need not be enumerated separately.

\begin{algorithm}[t]
\KwIn{Undirected tree $T=(V,E)$, with a vector of vertex weights
  $\mathbf{a} \in \mathbb{R}_+^n$}
\KwOut{A feasible vector of weights $\mathbf{x} \in \mathbb{R}_+^n$
  for $V$}
\BlankLine

\SetKwBlock{Begin}{begin}{end}

\SetKwFunction{KwPush}{Push-Path}
\SetKwFunction{KwImproveNode}{ImproveSubtree}
Let $v_1,v_2, \dots, v_{n-1}, v_n$ be the vertices in $T$ sorted in
post-order. \\

\For{$v \leftarrow 1$ to $n$}{
$x_v=max\{\sum_{u\in C(v_i)}x_u, a_v\}$\\%, \mathbf{y} \leftarrow \mathbf{x}$ \\
\KwImproveNode{$v$}
}
\BlankLine
\KwImproveNode{Vertex $u$} \\
\While{$\exists$ path $P$ from $u$ down to a vertex $v$, and $\epsilon>0$ such that $v$ is either a leaf or $x_v > \sum_{w \in C(v)} x_w$ and
  \KwPush($\mathbf{x}, P$,$\epsilon$)=$\epsilon$}{

\KwPush{$\mathbf{x}, p, \epsilon$} \\
%$\mathbf{y} \leftarrow \mathbf{x}$ \\ 
}

\BlankLine
\KwPush{Assignment $\mathbf{x}$, Path $P$, Non-negative real-value $\epsilon$} \\
\Begin{
Let $v_1,\ldots,v_k$ be the sequence of nodes on the $P$ from top to
bottom. \\
% \For{$i=1$ to $k$}{
% $x'_{v_i}=x_{v_i}$ 
% }

$old= \sum_{1 \leq i \leq k}|x_{v_i}-a_{v_i}|$ \\
$x_{v_1}=x_{v_1}-\epsilon$

\For{$i=2$ to $k$}{
$t=\sum_{u\in C(v_{i-1})}x_u-x_{v_{i-1}}$ \\
\If{$t > 0$}{
$x_{v_i}=x_{v_i}-t$
}
}
$new= \sum_{1 \leq i \leq k}|x_{v_i}-a_{v_i}|$ \\
\Return $old-new$
}
\caption{Push-Improve}\label{alg1}
\end{algorithm} 

The following theorem states that the output of Algorithm~\ref{alg1} is optimal.

\begin{thm}
\label{thm:main}
When Algorithm~\ref{alg1} terminates, the obtained vector $\mathbf{x}$
is a feasible and optimal assignment for the given tree $T$.
\end{thm}

Our proof of Theorem~\ref{thm:main} will proceed as follows. We begin by
characterizing the necessary push-path improvement at each step of the
while-loop. We then
inductively argue that before and after each push operation, the
value of the objective function for each sub-tree rooted in a child of the
current node remains optimal. We conclude by using an LP duality argument in order to
show that once no more push operations exist for the current vertex in the
for-loop, $T_v$ is assigned optimal $x$ values.

The following lemma refers to the series of improvements performed on
node $v$, and can be viewed as the set of invariants of the outer for-loop. %Additionally, we will refer to the improvement path $P$ that begins at $v$ and ends in some leaf in $T_v$.

\begin{lemma}
\label{lem1}
Let $v$ be the current node, $P=(v=u_0,\ldots,u_k)$ be a push-path such that for $1 \leq i \leq k$, $u_i \in C(u_{i-1})$. Then the following invariants hold throughout the execution of the inner while-loop:
\begin{enumerate}
\item
  If, for $\epsilon > 0$, $Push-Path(\mathbf{x}, P, \epsilon) > 0 $, then
  $Push-Path(\mathbf{x}, P, \epsilon) \leq \epsilon$. Furthermore, if for path $P$ and $\epsilon > 0$,
  $Push-Path(\mathbf{x}, P, \epsilon)=\delta > 0$, then there exists
  $\epsilon'>0$ such that $Push-Path(\mathbf{x}, P,\epsilon')=\epsilon'$.
\item If for path $P$ and $\epsilon>0$ $Push-Path(\mathbf{x}, P,
  \epsilon)=\epsilon$, then for each $u \in C(v)$,
  $T_u$ is optimally set before and after running
  $Push-Path(\mathbf{x}, P, \epsilon)$.
\end{enumerate}
\end{lemma}
\begin{proof}
  First, notice that the above invariants clearly hold if the current node $v$ is a leaf, as their initial $x$ values are set to their $a$ values, and will only be modified as a result of performing $Push-Path$ on their ancestors.
Assume that the invariants hold for all nodes preceding $v$ in the
post-order, and
suppose for contradiction that there exists some path
$P=(v=u_0,\ldots,u_m)$ and $\epsilon>0$ such that
$Push-Path(\mathbf{x}, P, \epsilon) > \epsilon$. \comment{Then since in
$Push-Path(\mathbf{x}, P, \epsilon)$ $x_v$ is reduced by $\epsilon$, it implies that an additional improvement is done on $T_{u_1}$ as well, thereby contradicting the optimality of $T_{u_1}$.}
The first part of the first invariant clearly holds since the sub-trees rooted in the children of $v$ are assumed optimal. Hence, any $\epsilon$-improvement on $v$ cannot entail an additional improvement on the rest of the push-path.

We now consider the second invariant, while briefly deferring the proof of the second part of the first invariant. 
First, notice that for each $\ell
\in C(v)-\{u_1\}$, the assignments to $T_{\ell}$ do not change. Let $P$ be a
modification-path, and $\epsilon>0$ such that $Push-Path(\mathbf{x},
P, \epsilon)=\epsilon$. On the other hand, notice that $x_v$ is
reduced by exactly $\epsilon$. This implies that $\| \mathbf{x}|_{T_{u_1}} -
\mathbf{a}|_{T_{u_1}} \|$ remains unchanged, thereby remaining optimal.

We now turn to the remaining part of the first invariant. Consider a modification-path $P$ 
and $\delta > 0$. By the first part of the invariant, $Push-Path(\mathbf{x}, P,\delta) \leq \delta$. 
If $Push-Path(\mathbf{x}, P,\delta)=\delta$, then the claim
holds trivially. Hence, assume $Push-Path(\mathbf{x}, P,\delta) < \delta$. 

We restrict ourselves
to dealing with $\delta$ values in the range $(0,x_v-a_v]$.
The following observation stems from the fact that during the
push operation, $x$ values along $P$ only decrease.

\begin{obs}
\label{obs0}
For path $P$ and $\epsilon>0$, if $Push-Path(\mathbf{x},P,\epsilon)>0$
\begin{equation}
  |\{j \in P: x_j > a_j \}| >  |\{j \in P: x_j \leq a_j \}|
\end{equation}
\end{obs}

In fact, using the induction hypothesis, we can make 
Observation~\ref{obs0} even stronger:
\begin{claim}
\label{cl0}
For path $P$ and $\epsilon>0$, if $Push-Path(\mathbf{x},P,\epsilon)>0$
\begin{equation}
  |\{j \in P: x_j > a_j \}| -  |\{j \in P: x_j \leq a_j \}|=1
\end{equation}
\end{claim}
Claim~\ref{cl0} can be justified by noticing that otherwise, the
sub-tree rooted in one of $v$'s children would be amenable to
path-improvements, contradicting optimality. The invariant follows, as we could simply set
$\epsilon$ to be the minimum (positive) amount that maintains the
number of nodes along $P$ with $x$ values that are larger than their
$a$ values.
\end{proof}

Lemma~\ref{lem1} implies that each push operation
improves the value of the objective function for the current
sub-tree, while maintaining the optimality of the sub-trees rooted in
the children of $v$. However, in order to show that the local optimum obtained by the
algorithm is the globally optimal feasible solution, we need
to argue that as long as the current assignment is not optimal, there
exists a feasible path-improvement with a corresponding $\epsilon>0$
value. The following theorem, which constitutes the main technical part of
this paper, formalizes this notion.

\begin{thm}
\label{thm:correctness}
Upon termination of the inner while-loop, the sub-tree rooted in vertex
$v$ is assigned optimal $x$ values.
\end{thm}
\begin{proof}

First, notice that
the algorithm clearly maintains the feasibility of the solution
throughout its execution.
The following observation follows from the definition of the
algorithm.

\begin{obs}
\label{obs1} 
During the execution of the algorithm, $x_v \geq
  a_v$. Furthermore, if $x_v = a_v$, the solution is trivially optimal. 
\end{obs}

The proof of Theorem~\ref{thm:correctness} will proceed as follows. We
give the LP for the optimization problem, and its corresponding dual
LP. We then construct a feasible solution for the dual LP that satisfies
the complementary slackness conditions with respect to the solution of
the algorithm. In order to construct a valid dual solution, we inductively
bootstrap the dual solutions constructed for the nodes rooted sub-trees.
From LP duality, we then conclude that the two solutions are optimal
for the primal and dual problems. Recall that we inductively assume
that the sub-trees rooted in the children of $v$ are optimally adjusted.

\medskip

It is not hard to write a linear program which formulates our problem.
This program and its dual can be seen below. The variables $d_i$ are
introduced to avoid using absolute values in the objective function.\\
% can be avoided by introducing an extra variable for bounding the difference $a_i-x_i$, for all $i
%\in T_v$)\\
%
%\noindent
\hspace{\stretch{1}}
\begin{subequations}
\begin{minipage}{0.45\textwidth}
\begin{alignat}{2}
\min\quad     \sum_{i \in T_v}d_i &     \nonumber \\
\text{subject to}\quad d_i + x_i    & \ge  a_i   \label{bound1}\\
 d_i - x_i &\ge  - a_i \label{bound2}\\
x_i - \sum_{j \in C(i)}x_j &\ge 0 \label{bound3}\\
    x_i     & \geq 0 &\qquad& \forall i \in T_v\label{bound3p}
\end{alignat}
\end{minipage}
\end{subequations}
%
%First note that since $d_i$ cannot be
%negative, we can omit the non-negativity constraint on $d_i$. As
%a result, this will allow a strict equality on its corresponding dual inequality.
%We define the variables $\lambda_i, \lambda'_i$, and $\alpha_i$ for 
%inequalities \ref{bound1}, \ref{bound2}, and \ref{bound3}, respectively.
%The corresponding dual LP would therefore be:
\hspace{\stretch{1}}
\begin{subequations}
\begin{minipage}{0.45\textwidth}
\begin{alignat}{2}
\max \quad    \sum_{i \in T_v}a_i(\lambda_i - \lambda'_i) &   \nonumber \\
\text{subject to}\quad   \lambda_i + \lambda'_i    & =  1 && \forall i \in T_v\label{bound4} \\
 (\lambda_i - \lambda'_i) + \alpha_i - \alpha_{p(i)} &\le  0 &&  \forall i \in T_v \backslash \{v\} \label{bound5}\\
 (\lambda_v - \lambda'_v) + \alpha_v &\le 0 \label{bound6} \\
    \lambda_i,\lambda'_i, \alpha_i & \geq 0 &\qquad& \forall i \in T_v
\end{alignat}
\end{minipage}
\end{subequations}
\hspace{\stretch{1}}
Note the special case for vertex $v$ (inequality \ref{bound6}). By denoting $\beta_i = \lambda_i - \lambda'_i$, one can simplify the dual LP:
\begin{subequations}
\begin{align}
\max \quad    \sum_{i \in T_v} a_i \beta_i  &     \nonumber \\
\text{subject to}\quad -1 \le \beta_i    & \le  1 \quad
&&\forall i \in T_v\label{bound7} \\
 \beta_i + \alpha_i - \alpha_{A(i)} &\le  0&& \forall i \in T_v-v \label{bound8}\\
\beta_v + \alpha_v &\le  0 \\
   \alpha_i & \geq  0 && \forall i \in T_v
\end{align}
\end{subequations}

We now summarize the necessary complementary slackness conditions
required by the dual:
\begin{align}
x_i &> a_i  \Rightarrow \lambda_i = 0, \lambda'_i = 1 
\quad (\beta_i=-1) \tag{C1} \label{slackness1} \\
x_i &< a_i \Rightarrow \lambda_i = 1, \lambda'_i = 0 
\quad (\beta_i=1) \tag{C2}\label{slackness2} \\
x_i &> \sum_{j \in C(i)}{x_j} \Rightarrow \alpha_i = 0 \tag{C3}
\label{slackness3} \\
x_i &> 0 \Rightarrow \lambda_i - \lambda'_i + \alpha_i -
\alpha_{p(i)} = 0 \tag{C4} \label{slackness4}
\end{align}
Since throughout the execution of the while loop $x_v \geq
a_v$ and the case where $x_v = a_ v$ is trivial, we will assume from
now on that $x_v > a_v$. This implies the last necessary condition:
\begin{equation}
x_v > a_v \Rightarrow \alpha_v = 1 \tag{C5}
\end{equation}

We begin by suggesting an initial assignment which might not be
feasible, and in addition, might violate one of the complementary
slackness properties.

The following lemma is a direct consequence of the construction of the
dual LP and the complementary slackness constraints. 
It refers to a family of assignments to the dual LP that satisfy 
a \emph{subset} of the complementary slackness conditions.

\begin{lemma}
\label{lem2}
Let $\mathbf{x,d}$ be a feasible solution for the primal such that the
sub-trees rooted in $v$ are optimally assigned and $v$ admits no
\emph{Push-Path} improvements. Let $\boldsymbol{\alpha,\beta}$
be an assignment for the dual variables such that the following holds:

\begin{equation}
\begin{matrix}
\left\{\begin{matrix}
\alpha_i = \alpha_{p(i)}-\beta_i, & \text{ if } x_i > 0\\ 
\alpha_i \leq \alpha_{p(i)}-\beta_i, & otherwise 
\end{matrix}\right. \;& \beta_i = \left\{\begin{matrix}
-1, & \text{ if } x_i > a_i \\ 
1, & \text{ if } x_i < a_i \\
\text{a value in }[-1,1], & \text{ if }x_i = a_i
\end{matrix} \right .
\end{matrix}
\end{equation}

% \begin{enumerate}
% \item $\alpha_i = \left\{\begin{matrix}
% \alpha_i = \alpha_{p(i)}-\beta_i, & \text{ if } x_i > 0\\ 
% \alpha_i \leq \alpha_{p(i)}-\beta_i, & otherwise 
% \end{matrix}\right.$

% \item $\beta_i = \left\{\begin{matrix}
% -1, & \text{ if } x_i > a_i \\ 
% 1, & \text{ if } x_i < a_i \\
% \text{a value from }[-1,1], & \text{ if }x_i = a_i
% \end{matrix}\right.$
% \end{enumerate}
Then $\boldsymbol{\alpha,\beta}$ satisfy all the properties of a feasible
dual solution, and $(\boldsymbol{\alpha,\beta})$ along with
$(\mathbf{x,d})$ satisfy complementary slackness \textbf{except} that
$\alpha_i$ might be negative for some nodes, and condition~\ref{slackness3}
could be falsified.
\end{lemma}

Next, we observe that if our modified dual LP admits an optimal
\emph{feasible} solution, then our range of possible values for
$\boldsymbol{\alpha,\beta}$ can be narrowed due the total unimodularity
of the simplified constraint matrix of the dual LP:
\begin{obs}
\label{obs2}
If the dual LP has an optimal and feasible solution, then it has an
integral, feasible and optimal solution, as well. In particular, for
every $i \in T_v$, $\beta_i \in \{-1,0,1\})$.
\end{obs}
Observation~\ref{obs2} can be verified by induction on the constraint
matrix of the dual LP, in order to show that every square sub-matrix
of it has a determinant of $\pm 1$.

The following lemma complements Lemma~\ref{lem2} by suggesting a
concrete assignment for each $\beta_i$ in the case whenever $x_i=a_i$.

\begin{lemma}
\label{lem3}
Consider an assignment as described in Lemma~\ref{lem2}. If we set $\beta_i = 1$ whenever $x_i =
a_i$, then:
\begin{equation*}
\forall j \in T_v, x_j > \sum_{k:\text{child of j}}x_k \Rightarrow
\alpha_j \leq 0
\end{equation*}
\end{lemma}
\begin{proof}

We prove the claim by way of contradiction. Suppose that the
claim is false, and let $j$ be the highest node for which the claim
does not hold. That is, $x_j > \sum_{k \in C(j)}x_k$ and
$\alpha_j > 0$. Consider $P_{j\rightarrow v}$, the path from $j$ to
$v$. As we are trying to prove an upper bound for
$\alpha_j$, we will assume that for every node $k$ on the path from
$v$ to $j$ $\alpha_k=\alpha_{p(k)} - \beta_k$, as lower values will
only strengthen our claim. This implies 
\begin{equation}
\alpha_j = - \sum_{k \in P_{j \rightarrow
    v}}\beta_k. 
\end{equation}
Since $\beta_i=1$ for all nodes $i$ such that $x_i \leq a_i$, and
$\beta_i=-1$ otherwise, $\alpha_j>0$ implies:
\begin{equation}
|\{k \in P_{j \rightarrow v}: x_i > a_i\}| > |\{k \in P_{j \rightarrow
  v}: x_i \leq a_i \}|
\end{equation}
This implies that we can reduce all $x$ values of nodes $P_{j \rightarrow v}$ 
by an amount of at most $x_j - \sum_{k \in C(j)}x_k$ 
so as to get a feasible solution with a better objective
function value. However, this is exactly a push operation, 
thereby contradicting the assumption of no further path-paths.
\end{proof}
The following corollary is the contra-positive statement of
Lemma~\ref{lem3}
\begin{cor}
\label{cor1}
If there exists a node $j \in T_v$ such that $\alpha_j >0$ and $x_j -
\sum_{k \in C(j)}x_k > 0$, then there exists an ancestor $i$ of
$j$ such that $\beta_i \in \{0, -1\}$ and $x_i = a_i$. 
\end{cor}

We now prove the main theorem by way of induction.
We inductively assume that the sub-trees rooted in $v$ have both an
optimal setting for the primal LP, and there exists an integral and feasible
solution for the dual LP that satisfy the
complementary slackness conditions.
Without loss of generality, we assume that
no child $i$ of $v$ has $x_i=0$, since otherwise, we could use its
assignments without any modifications, as $x_i$ does not harden the
feasibility constraints of $v$.

Consider the assumed set of assignments for the sub-trees rooted in
$v$. By the assumption, they have corresponding assignments to the
dual LPs. Observe that since the conditions listed in
Lemma~\ref{lem2} are a subset of the complementary slackness
conditions, Lemma~\ref{lem2} applies to them automatically.

We will start from a tentative solution to the dual by
 initially set the ($\boldsymbol{\alpha, \beta})$ 
according to the assumed assignments, and set
$\alpha_v=1,\; \beta_v=-1$. We let $\mathbf{s_1}$ denote the above
assignment. Notice that for each child $i$ of $v$, the dual LP that corresponds 
to the current assignment had
\begin{equation*}
\alpha_i + \beta_i = 0,
\end{equation*}
as $i$ was the root (this is a strict equality 
as by our assumption $x_i>0$). However, 
in the current LP, the corresponding dual inequality becomes
\begin{equation*}
\beta_i + \alpha_i - \alpha_v= 0,
\end{equation*}
As $\alpha_v=1$, this equality is therefore violated. In order to
rectify this, we first raise all the $\alpha$ value (except
$v$'s) by $1$, and denote the resulting solution by $\mathbf{s_2}$.
Note that by the feasibility of the original assignments to the
sub-trees and by the definition of $\mathbf{s_2}$ all the nodes in
$T_v$ have non-negative $\alpha$ values.
Also observe that $\mathbf{s_2}$ now has all the properties listed in
Lemma~\ref{lem2}. Thus, by Lemma~\ref{lem2}, we can conclude that
$\mathbf{s_2}$ is a feasible solution to the dual LP, and
$\mathbf{s_2}$ along with $\mathbf{(x,d)}$ satisfy complementary
slackness except that complementary slackness condition~\ref{slackness3} might be violated.

Our next step would be to adjust $\mathbf{s_2}$ so as to fix any
violation condition of condition~\ref{slackness3}. Let $W$ be the set of all
infeasible nodes:
\begin{equation}
W = \{ j \in T_v: x_j > \sum_{k \in C(j)}x_k \text{ and } \alpha_i > 0\}
\end{equation}
By Corollary~\ref{cor1}, for each $j \in W$ there exists an ancestor $i$ such that (1)
$x_i=a_i$ and (2) $\beta_i \in \{-1,0\}$. We let 
\begin{equation}
X = \{i \in T_v:
x_i = a_i, \beta_i \in \{-1,0\}\}
\end{equation}
Moreover, we let 
\begin{equation}
Y=\{i \in X:
\text{there is no ancestor of $i$ in $X$}\}
\end{equation}
Thus, for each node $j \in W$ there exists an ancestor $i \in Y$. 

We now define the final solution to the dual LP. Define assignment $\mathbf{s_3}$ to the dual LP for $T_v$ by taking
solution $\mathbf{s_2}$ with the following modifications:
\begin{enumerate}
\item $\forall k \in T_i$, such that $i \in Y$, subtract $\alpha_k$
  by $1$. \label{step1}
\item $\forall i \in Y$ add 1 to $\beta_j$. \label{step2}
\end{enumerate}
Increasing the $\beta_j$ values by 1 makes sure that complementary
slackness condition~\ref{slackness4} is satisfied after applying the first
step. Applying the first modification step guarantees that
complementary slackness condition~\ref{slackness3} is again satisfied, as all
nodes in $W$ undergo the first modification. Observe that by
definition, all the sub-trees rooted in nodes in $Y$ are pair-wise
disjoint. Hence,
each $\alpha$ value can be decremented at most once. Also observe that
in addition to nodes in $W$, other nodes may have their $\alpha$
values decremented. However, as by the definition of $W$, these nodes
do not need to maintain condition~\ref{slackness3}, and thus this step will not
violate their constraints. In addition, their $\alpha$ values are
guaranteed to remain non-negative as they were previously incremented
by 1.

In conclusion, all of the complementary slackness conditions for the
dual LP now hold for $\mathbf{(x,d)}$ and $\mathbf{s_3}$. Therefore,
$\mathbf{(x,d)}$ is an optimal solution for $T_v$.
\end{proof}

%%% Local Variables: 
%%% mode: latex
%%% TeX-master: "Regression"
%%% End: 

\section{The Algorithm}\label{sec:algfinal}
Given the general technique presented in Algorithm~\ref{alg1}, we conclude
our results by giving an $O(n^2)$ algorithm that follows the spirit
of improving by pushing the surplus from a given vertex downward, along a
path. Recall that the algorithm \textit{Push-Improve} performs the
push operations \emph{one} path at a time. Instead, we can
leverage the fact that some paths can share the same
prefix. Specifically, instead of the inner while-loop,
executed for each node $v$ in the tree, 
we introduce a depth-first-search algorithm in which for 
each node $j \in T_{v}$, the algorithm remembers the maximal
amount, pushable  through $P_{v\rightarrow j}$. %Clearly, this measure
%will depend on the different paths that go through $j$.

We make use of two measures, defined for each node $u \in
T_{v}$. Let $\delta_u = |\{ j \in P_{v \rightarrow u}: x_j > a_j \}| -
|\{ j \in P_{v \rightarrow u}: x_j \leq a_j \}|$. In other words, for
any push operation along $P_{v \rightarrow u}$,
$\delta_u$ is the difference between the number of nodes that will
improve the objective function value, and the number of nodes that
will worsen the objective function value, if we push a small enough
value through $P_{v \rightarrow u}$. Additionally, we define
the positive bottleneck along $P_{v \rightarrow u}$ as $\epsilon_u =
\min_{j \in P_{v \rightarrow u}}\{ x_j - a_j: x_j > a_j
\}$. This is the maximum $\epsilon$ we can push on the path $P_{v \rightarrow u}$ while
gaining exactly $\delta_u \epsilon$ in the objective function. 
%Intuitively, this value gives the number of ``good'' nodes $u \in T_{v}$ ($x_u > a_u$) along $P_{v  \rightarrow u}$ for which their $x_u$ values can be decreased, without flipping the sign of $x_u - a_u$. 
In order to maintain feasibility, we restrict $\epsilon_u$ to be \emph{no more} than $x_k$, for any node $k$ on
$P_{v\rightarrow u}$. 
This value will have a similar function as the $\epsilon$ value
given in Algorithm~\ref{alg1}. That is, for the current node $v$, and
a successor $u$, $\epsilon_u$ will serve as the amount of
\emph{excess} we push through $P_{v \rightarrow u}$. Our algorithm will maintain feasability by
restricting the decrease in $x_u$ by the sum of the decreases made on its direct
children of $u$ (unless $x_u$ was strictly bigger than the sum of the $x$'s of its children before the decrease.)

The final algorithm for optimizing the assignment to $T_v$, can
be seen as Algorithm~\ref{alg2} in Appendix~\ref{app:1}.
It differs from Algorithm~\ref{alg1} in the way the sub-tree
$T_{v}$ is modified for each node $v$. %As before, it assumes that the
%sub-trees rooted in node $v$'s children are optimally configured.

%We now justify the correctness of the algorithm.
The following theorem states that Algorithm~\ref{alg2} is optimal.
\begin{thm}
\label{thm:dfs}
When Algorithm~\ref{alg2} leaves node $v \in V$, there is no push-path
going from the root $r$, ends at a leaf, and passes through $v$.
\end{thm}
\begin{proof}
First, we note the following observation, which suggests that the
potential for improvement on any path
from the root to a node $v$ cannot increase.
\begin{obs}
Let $u$ be a node in $T$. Let
\begin{equation*}
\delta^*_u = |\{i \in P_{r\rightarrow u}: x_v > a_v \}| - |\{i \in P_{r\rightarrow u}: x_v \leq a_v \}|.
\end{equation*} 
\label{obs:dfs}
$\delta_u^*$ does not increase throughout the execution of the algorithm.
\end{obs}
The observation follows immediately from the fact that the only
modifications to the $x$ values of the nodes are decreases.

We proceed to prove the lemma by induction on the height $h$ of the node
$v$. For $h=0$ (leaves), the claim is trivial. %as initially they are optimally set. 
Assume the claim holds for $h=k$, and let $v$ be a
node of height $k+1$. The claim follows immediately from Observation~\ref{obs:dfs}: no sub-tree rooted in a
child of $v$ can be improved as a result of a push-path through
it. Additionally, the path from $r$ to $v$ never becomes amenable to
improvements through push operations, once the algorithm leaves $v$. This concludes
the proof.
\end{proof}

\paragraph{Running Time} The algorithm essentially performs a
depth-first-search for every node $v$ on the tree. %, where a constant number of operations is done when visiting each node. 
Therefore, the running time of the algorithm is $O(n^2)$.
%%% Local Variables: 
%%% mode: latex
%%% TeX-master: "Regression"
%%% End: 

%\input{Regression-L2.tex}
\section{Conclusions and Future Work}\label{sec:conclusion}
We have demonstrated the technical difficulties that our problem entails, as
well as an efficient method for handling a broad class of instances of the problem.
Due to their high efficiency, our methods can be run on relatively large instances in practice.
%also regarded as practical, in additional to being theoretically interesting. 
We also believe that our algorithm might be applicable to settings beyond recommendation systems.
%Also, we believe that the studied setting might be relevant to other settings
%in which some sort of smoothing step on graphs is necessary.

An immediate open question is to extend our algorithm to the case of general
DAGs. It seems that one needs some new ideas to give a combinatorial algorithm
for this general case. In fact even a (fast) approximation algorithm for this
case seems to be beyond the reach of our techniques.
%approach used for the special case of trees does not, in general, lend itself to such substantially richer settings. 
%This case seems to be as acute even for a more restricted
%case of layered graph. Indeed, even a (constant factor) approximation
%algorithm would be quite interesting for this case.
Another interesting direction would be to consider other 
measures such as the $\ell_2$-norm. Due to the fundamental difference
between the $\ell_1$ and $\ell_2$ norms, we suspect that this
different distance measure will require a completely different approach.

In addition to considering alternative objective functions, we can also
consider other constraints. For instance, we can consider
comparing the value assigned to each node to the \emph{average} value of its
children. %, or perhaps to the $k$th moment of its children's values.
Another type of constraint would be to require equality between
the value of a node, and the sum of the values of its children.

%%% Local Variables: 
%%% mode: latex
%%% TeX-master: "Regression"
%%% End: 

%\nocite{ahkw,bbbb,gcbh,ze}
%\nocite{ahkw,gcbh,ze}
%\bibliographystyle{abbrv}
\bibliographystyle{alpha}
\bibliography{Regression}  

\newcommand{\etalchar}[1]{$^{#1}$}
\begin{thebibliography}{MJDP{\etalchar{+}}00}

\bibitem[AHKW06]{ahkw}
Stanislav Angelov, Boulos Harb, Sampath Kannan, and Li-San Wang.
\newblock Weighted isotonic regression under the $\ell_1$ norm.
\newblock In {\em SODA}, pages 783--791, 2006.

\bibitem[BBBB72]{bbbb}
R.~E. Barlow, D.~J. Bartholomew, J.~M. Bremmer, and H.~D. Brunk.
\newblock {\em Statistical Inference Under Order Restrictions}.
\newblock Wiley, 1972.

\bibitem[BC90]{best}
M.~J. Best and N.~Chakravarti.
\newblock Active set algorithms for isotonic regression: a unifying framework.
\newblock {\em Math. Program.}, 47:425--439, August 1990.

\bibitem[CKP07]{ckp}
Deepayan Chakrabarti, Ravi Kumar, and Kunal Punera.
\newblock Page-level template detection via isotonic smoothing.
\newblock In {\em WWW}, pages 61--70, 2007.

\bibitem[DCZ{\etalchar{+}}10]{dczmbzbld}
Anlei Dong, Yi~Chang, Zhaohui Zheng, Gilad Mishne, Jing Bai, Ruiqiang Zhang,
  Karolina Buchner, Ciya Liao, and Fernando Diaz.
\newblock Towards recency ranking in web search.
\newblock In {\em WSDM}, pages 11--20, 2010.

\bibitem[DKK11]{dkk11}
Gideon Dror, Noam Koenigstein, and Yehuda Koren.
\newblock Yahoo! music recommendations: Modeling music ratings with temporal
  dynamics and item taxonomy.
\newblock In {\em Recommender Systems}, Chicago, IL, USA, 2011. ACM.

\bibitem[Fei98]{feige98}
Uriel Feige.
\newblock A threshold of ln n for approximating set cover.
\newblock {\em J. ACM}, 45:634--652, July 1998.

\bibitem[Fle04]{fleischer}
Lisa Fleischer.
\newblock A fast approximation scheme for fractional covering problems with
  variable upper bounds.
\newblock In {\em Proceedings of the fifteenth annual ACM-SIAM symposium on
  Discrete algorithms}, SODA '04, pages 1001--1010, Philadelphia, PA, USA,
  2004. Society for Industrial and Applied Mathematics.

\bibitem[FT84]{friedman}
J.~Friedman and R.~Tibshirani.
\newblock The monotone smoothing of scatterplots.
\newblock {\em Technometrics}, 1984.

\bibitem[KFB09]{kfb}
R\'{e}mon Kamp, Ad~Feelders, and Nicola Barile.
\newblock Isotonic classification trees.
\newblock In {\em Proceedings of the 8th International Symposium on Intelligent
  Data Analysis: Advances in Intelligent Data Analysis VIII}, IDA '09, pages
  405--416, Berlin, Heidelberg, 2009. Springer-Verlag.

\bibitem[LB01]{leuraud}
Klervi Leuraud and Jacques Benichou.
\newblock A comparison of several methods to test for the existence of a
  monotonic dose-response relationship in clinical and epidemiological studies.
\newblock {\em Statistics in Medicine}, 20(22):3335--3351, 2001.

\bibitem[MAC01]{mancuso}
Jessica~Y. Mancuso, Hongshik Ahn, and James~J. Chen.
\newblock Order-restricted dose-related trend tests.
\newblock {\em Statistics in Medicine}, 20(15):2305--2318, 2001.

\bibitem[MJDP{\etalchar{+}}00]{morton}
Tony Morton-Jones, Peter Diggle, Louise Parker, Heather~O. Dickinson, and Keith
  Binks.
\newblock Additive isotonic regression models in epidemiology.
\newblock {\em Statistics in Medicine}, 19(6):849--859, 2000.

\bibitem[MSCZ10]{mscz}
Taesup Moon, Alex~J. Smola, Yi~Chang, and Zhaohui Zheng.
\newblock Intervalrank: isotonic regression with listwise and pairwise
  constraints.
\newblock In {\em WSDM}, pages 151--160, 2010.

\bibitem[PG08]{pg}
Kunal Punera and Joydeep Ghosh.
\newblock Enhanced hierarchical classification via isotonic smoothing.
\newblock In {\em WWW}, pages 151--160, 2008.

\bibitem[RWD88]{rwd}
T.~ROBERTSON, F.~T. Wright, and R.~L. Dykstra.
\newblock {\em Order Restricted Statistical Inference}.
\newblock Wiley, 1988.

\bibitem[Sto]{stoutxxisoregalg}
Quentin~F. Stout.
\newblock {Isotonic Regression Algorithms}.
\newblock http://www.eecs.umich.edu/\~{}qstout/IsoRegAlg.html.
\newblock Retrieved: Aug. 6th, 2011.

\bibitem[Sto08]{s1}
Quentin~F. Stout.
\newblock Unimodal regression via prefix isotonic regression.
\newblock {\em Computational Statistics \& Data Analysis}, 53(2):289--297,
  2008.

\bibitem[Ulm86]{ulm}
K.~Ulm.
\newblock Nonparametric analysis of dose-response relations in epidemiology.
\newblock {\em Mathematical Modelling}, 7(5-8):777 -- 783, 1986.

\end{thebibliography}
\appendix
\section{The Weighted $\ell_1$-Norm Case}
\label{sec:weighted}

We now discuss the case where the nodes on the given tree can
have varying levels of importance with respect to the objective function. Specifically,
as done in related studies, we consider the case in which for each node $i \in V$, there is an associated weight
$w_i$. For ease of presentation, we assume that all weights are integral, i.e. $\mathbf{w} \in \mathbb{N}_{\geq}^{n}$. The objective function will be
$g(\mathbf{x}) =\sum_{i \in V} w_i \cdot |a_i - x_i|$.

Hence, we can simply reinterpret the variable 
$\delta_i$ defined for Algorithm~\ref{alg2} as the weighted 
balance between the nodes which will benefit, and the nodes 
that will ``suffer'' as a result of a Push-Path operation. More precisely,
when considering the path $P_{v \rightarrow i}$, we will compute 
$\delta_i = \sum_{j \in P_{v \rightarrow i}:x_j>a_j}w_j - \sum_{j \in P_{v \rightarrow i}:x_j \leq a_j}w_j$.
Therefore, a feasible push-improvement across a path $P_{v \rightarrow u}$ would lead to an improvement if and only if $\delta_u>0$. 

The above discussion leads to the following simple modification to
procedure \texttt{SetParams}, given in Algorithm~\ref{alg3}.
\begin{algorithm}[ht]
\KwIn{Vertex $v$}
\BlankLine
\SetKwBlock{Begin}{begin}{end}
\SetKwFunction{KwSetParams}{Set-Params}
\label{alg3}
  \KwSetParams{Vertex $i$, Non-negative integer $\delta$, Non-negative
  real-value $\epsilon$} \\
\Begin{
  \If{$x_i > a_i$}{
    $\delta_i \leftarrow \delta + w_i, \epsilon_i \leftarrow \min\{\epsilon, x_i - a_i\}$
  }
  \Else{
    $\delta_i \leftarrow \delta -w_i, \epsilon_i \leftarrow \min \{x_i, 
    \epsilon \}$ \\
  }
}
\caption{The modified \texttt{Set-Params} procedure for the weighted case}
\end{algorithm} 
Note that the weighted case reduces to the unweighted case by setting the weights to $1$.
Clearly, the algorithm has the same $O(n^2)$ running time of the original algorithm. 
The following theorem argues about the optimality of the modified algorithm:
\begin{thm}
The algorithm resulting from the modification given in Algorithm~\ref{alg3} obtains the optimal
weighted-$\ell_1$ objective function value. 
\end{thm}
\begin{proof}
In order to argue
about the correctness of the modified algorithm, we compare the objective function obtained by the algorithm
to the one obtained by the original algorithm, on an equivalent unweighted tree.

\paragraph{The construction} We construct the tree $\tilde T=(\tilde V, \tilde E)$ by 
replacing each node $i$ with a chain $i_1,\ldots,i_{w_i}$, 
such that for any $1 \leq j < w_i$, $(i_j, i_{j+1}) \in \tilde E$. Additionally, we set for
children $k \in C(i)$ $(k_{w_k}, i_{1}) \in \tilde E$, and for $i$'s parent $\ell$ $(i_{w_i}, \ell_1) \in \tilde E$.

It is easy to see that $\tilde T$ is a tree. 
Notice that $\tilde T$ might be arbitrarily large (according to the weights). 
However, it is used only for the sake of proof of correctness, and never actually constructed
by the algorithm.
The following proof sketch highlights the equivalence of the uniform weight
case to the weighted case.
\begin{claim}
\label{cl4}
Let $\mathbf{x}$ and $\mathbf{\tilde x}$ be optimal assignments for $T$ and $\tilde T$, respectively. Then
$g(\mathbf{x}) = f(\mathbf{\tilde x})$. 
\end{claim}
The following immediate observation, which follows from the construction of $\tilde T$, implies the above claim.
\begin{obs}
\label{obs4}
Let $\mathbf{\tilde x}$ be an optimal assignment for $\tilde T$. 
Then for any chain $(i_1, \ldots, i_{w_i})$ that corresponds to vertex $i$ in $T$:
\begin{equation*}
x_{i_1} = x_{i_2} = \ldots = x_{i_{w_i}} 
\end{equation*}
\end{obs}
The following claim complements claim~\ref{cl4}:
\begin{claim}
\label{cl5}
Let $\mathbf{x}$ and $\mathbf{\tilde x}$ be the feasible assignments returned by Algorithm~\ref{alg2} and
the modified algorithm for weighted trees, respectively. Then
\begin{equation*}
g(\mathbf{x}) = f(\mathbf{\tilde x})
\end{equation*}
\end{claim}
\end{proof}

%%% Local Variables: 
%%% mode: latex
%%% TeX-master: "Regression"
%%% End: 

\section{The $\ell_{\infty}$-norm}\label{sec:ellinfty}
We now turn our attention to the case of the $\ell_{\infty}$-norm; i.e. minimizing
the maximal difference $max_{u \in V}|a_i - x_i|$. In contrast to the case of the $\ell_1$-norm,
this optimization problem can be solved in a straightforward manner by using dynamic programming, even
when the underlying graph is a directed acyclic graph.

For a given value $t \geq 0$, the algorithm will go over all nodes and tries to produce an assignment of objective value
at most $t$. We can show that if the algorithm fails then there is no valid assignment of objective value at most $t$.
To find the optimal objective value then one only needs to run a binary search on the variable $t$.

\begin{algorithm}
\label{alg-infty}
\KwIn{DAG G=(V,E), vertices $1,\ldots,n$ sorted in topological order, vertex weight vector $\mathbf{a}$.}

\SetKwBlock{Begin}{begin}{end}

\For{$i \leftarrow 1 \text{ to } n$}{
$x^{min}_i \leftarrow \max\{0,  \sum_{j \in C(i)}x^{min}_j, a_i - t\}$ \label{line:alg-infty:maxLine}
}
\Return{ $\mathbf{x}^{min}$}
\caption{The dynamic programming algorithm for the $\ell_{\infty}$-norm case}
\end{algorithm}

As mentioned above, we perform a binary search on $t$ in the range $[0,\sum_i a_i]$.
Clearly, for an instance of the problem with optimal solution value $\tau$, 
the running time of Algorithm~\ref{alg-infty} would be $O(n \cdot log \tau)$. We now briefly outline the proof of correctness
of the algorithm.

\begin{thm}
For any given $t \geq 0$, $\mathbf{x}=\mathbf{x}^{min}$ is a valid solution. Furthermore, if
$||\mathbf{x}-\mathbf{a}||_{\infty} > t$, then there does not exist a valid solution $\mathbf{x'}$ such that
$||a-x'||_{\infty} \leq t$.
\end{thm}
\begin{proof}
  The validity of $\mathbf{x}$ follows from definition. To prove the second part we show the following simple lemma.
  \begin{lemma}
    If $\mathbf{x'}$ is a valid solution and
    $||\mathbf{a}-\mathbf{x'}||_{\infty} \leq t$, then for all $i$,
    $x'_i \geq x^{min}_i$
  \end{lemma}
  \begin{proof}
    The proof follows with a simple induction on $i$. Note that because $||\mathbf{a}-\mathbf{x'}||_{\infty} \leq t$, $x'_i \geq a_i - t$.
    Furthermore, $\mathbf{x'}$ is a valid solution so $x'_i \geq 0$ and
    \begin{align*}
      x'_i &\geq \sum_{j \in C(i)} x'_j \geq \sum_{j \in C(i)} x^{min}_j = x^{min}_i,
      \intertext{where we have used the induction hypothesis for the second inequality. It then follows that,}
      x'_i &\geq \max\{0,  \sum_{j \in C(i)}x^{min}_j, a_i - t\} = x^{min}_i.\qedhere
    \end{align*}
  \end{proof}
Now assume that there is a valid solution $\mathbf{x'}$ with objective value at most $t$. It follows that for all $i$,
$$ a_i - t \leq x^{min}_i \leq x'_i \leq a_i + t,$$
that is, $||\mathbf{x}^{min} - \mathbf{a}||_\infty \leq t$.
\end{proof}

%%% Local Variables: 
%%% mode: latex
%%% TeX-master: "Regression"
%%% End: 

\section{Hardness of Approximation in General Graphs}\label{sec:hardness}
As mentioned before when the objective value is the $\ell_1$ norm of the difference between the ${\bf x}$ and ${\bf a}$ vectors the
(most general case of the) problem can be solved exactly in polynomial time by solving a linear program. In fact, it is not hard to
see that using a similar approach one can solve this general case of the problem for \emph{any} $\ell_p$ norm with
$1 \leq p \leq \infty$. The only difference is that one has a linear program with infinitely many facets which has an efficient
separation oracle and can be solved with the Ellipsoid method. 

For an instance where all the input values are integral, one might ask whether the task of finding an
optimal \emph{integral} solution is tractable or not. This is especially interesting for the $\ell_1$ case, since in the case of
trees, an integral solution can be found efficiently by our algorithm
of Section~\ref{sec:algabs}, if the initial $a$ values are integrals.
Unfortunately, as soon as one considers the DAG case (even the special case of layered dags) this problem becomes intractable for
essentially all $\ell_p$ norms. The following theorem summarizes our hardness results.

\begin{thm}
Unless $NP \subseteq TIME(n^{O(\log\log n)})$ it is NP-hard to approximate the Integral Isotonic Regression problem for the case of
directed acyclic graphs better than $\Theta(\left(\log n\right)^{1/p})$ for the $\ell_p$ norm.
\end{thm}
\begin{proof}
We prove the theorem by a reduction from the \textit{Set Cover} problem. In the Set Cover problem one is given sets $S_1, S_2, \ldots, S_m$ such
that $S_1 \cup S_2 \cup ... \cup S_m = \{1, 2, \ldots, n\}$ and the objective is to select a minimum number of $S_i$'s such that their union is
still $\{1, 2, \ldots, n\}$. It is a well known result of Feige~\cite{feige98} that unless $NP \subseteq TIME(n^{O(\log\log n)})$ it is NP-hard to
approximate Set Cover better than a factor of $(1-o(1))\log n$. Our reduction uses vertex weights so as to simplify
the construction. However, one can easily adapt the construction to the uniform case by adding multiple copies of nodes so as to simulate
large weights.

Given an instance of the Set Cover problem We construct the following
instance of the SBHSP problem:
\begin{itemize}
\item The vertex set of the output digraph will be $V = \{ v_1, \ldots, v_m, u_1,\ldots u_n\}$.
\item The edge set of the output digraph will be $E = \{ (v_i, u_j) : j \in S_i \}$.
\item The $a$ values on the vertices will be as follows. For all $v_i$ we have $a(v_i) = 1$, while for all $u_j$ we have
  $a(u_j) = |\{i : j \in S_i\}|-1$.
\item The $w$ values (weights) of the vertices will be as follows. For all $v_i$ we have $w(v_i) = 1$, while for all $u_j$ we
  have $w(u_j) = m$.
\end{itemize}
On the one hand it is easy to see that for any set cover of the original instance (of size $\alpha$) one can construct a solution to the
SBHSP instance (of cost $\sqrt[p]{\alpha}$) by assigning $x(v_i) = 0$ if $S_i$ is selected and $1$ otherwise, and $x(u_j) = a(u_j)$
for all $u_j$.

On the other hand it is not hard to see that the optimal solution to the SBHSP will have $x(v_i) \in \{0, 1\}$ for all $i$
and $x(u_j)=a(u_j)$ for all $j$. Furthermore, for any such solution (of cost $\alpha$) the set of $S_i$ for which $x(v_i)=0$ can be easily seen to
be a valid set cover of size $\alpha^p$. Hence, a hardness of approximation of $(1-o(1))\log n$ for the Set Cover problem implies a hardness of
approximation of $\Theta(\sqrt[p]{\log n})$ for SBHSP when the objective value is defined using the $\ell_p$ norm for any
$1 \leq p < \infty$.
\end{proof}

\begin{remark}
  The hard-instances of Set-Cover generated by Feige~\cite{feige98} have less Sets than elements. As a result it is not hard to see that the
  hardness achieved by the proof of the above theorem is in fact $((1-o(1))\ln n)^{1/p}$ for the weighted case and $\left(\frac{(1-o(1))\ln n}{2}\right)^{1/p}$ for the non-weighted case.
\end{remark}

%%% Local Variables: 
%%% mode: latex
%%% TeX-master: "Regression"
%%% End: 

\section{FPTAS for optimizing under the $\ell_1$ norm for Bilayered graphs}\label{sec:bilayer}
Consider a DAG $G = (V, E)$ which is bilayered, i.e. the vertex set can be partitioned as $V = U \cup W$ and each edge is from the $U$ side to
the $W$ side ($E \subseteq U \times W$.) In this section we show a fast Fully Polynomial Approximation Scheme for SBHSP with the $\ell_1$ norm
for such DAGs. The run time will be close to linear in the size of the DAG. The algorithm is a simple reduction to a well known class of problems
which admit such FPTASes. These problems are restricted class of the Mixed Positive Packing and Covering Problem, see \cite{fleischer}.

We start by the following simple observation.
\begin{lemma}\label{lem:bilayerobs}
  When optimizing the $\ell_1$ norm and when the input DAG is bilayered there is always an optimal solution with the following two properties,
  (i) $\forall w \in W : x_w = a_w$, (ii) $\forall u \in U : x_u \leq a_u$.
\end{lemma}
\begin{proof}
  Consider any optimal solution ${\bf x}$, and a vertex $w \in W$ for which $x_w \neq a_w$. If $x_w < a_w$ changing the assigned value of this
  vertex to $a_w$ produces another valid solution with a better
  objective function value. Now consider the case in which $x_w > a_w$. If $x_w > \sum_{u \in C(w)} x_u$
  then again we can improve the objective function by decreasing $x_w$ slightly, and if $x_w = \sum_{u \in C(w)} x_u$ we can simultaneously decrease
  $x_w$ and the assigned value of some of its children. This last step
  would help the objective value due to the improvement on $w$, and
  possibly hurt it by the exact same
  amount due to the decrease on its children, while maintaining a
  valid the solution. Doing this step on every node on the $W$ side results in a solution that satisfies the
  first condition.

  For the second condition observe that if $x_u > a_u$ we can simply decrease it to $a_u$ without changing the validity of the solution while
  decreasing the objective function. In other words any optimal solution must satisfy the second condition.
\end{proof}
Given the above lemma one can write the following linear program whose solution is the exact value of the optimal solution. The left hand side is
the original program based on the LP \eqref{bound1}-\eqref{bound3p} while the right hand side is the result of a simplification.\\
\hspace{\stretch{1}}
\begin{subequations}
\begin{minipage}{0.35\textwidth}
\begin{alignat*}{2}
\min\quad     \sum_{v \in V}d_i\\
\text{subject to}\quad d_u    & \ge  0  &\quad& \forall u \in U\\
x_u    & \ge  0  && \forall u \in U\\
d_u + x_u    &=  a_u && \forall u \in U\\
\sum_{u \in C(w)} x_u &\le a_w &&\forall w \in W
\end{alignat*}
\end{minipage}
\end{subequations}
\hspace{\stretch{1}}
\begin{subequations}
\begin{minipage}{0.6\textwidth}
\begin{alignat}{2}
\min\quad     \sum_{v \in V}d_i\nonumber\\
\text{subject to}\quad d_u    & \ge  0  &\quad& \forall u \in U \label{lp:bilayer:1}\\
d_u    &\leq  a_u && \forall u \in U \label{lp:bilayer:2}\\
\sum_{u \in C(w)} d_u &\ge (\sum_{u \in C(w)} a_u) - a_w &&\forall w \in W \label{lp:bilayer:3}
\end{alignat}
\end{minipage}
\end{subequations}
\hspace{\stretch{1}}
Once written in this form the above formulation is a, so called, Mixed Positive Packing and Covering Program. In fact it is among a certain
class of such programs for which Fleischer~\cite{fleischer} provides a fast FPTAS. In particular, we have the following theorem.
\begin{thm}
  When the input is a bilayered graph and the objective value is in terms of the $\ell_1$ norm, there is an algorithm that given $\epsilon>0$ runs
  in time $O(|V||E|\log(|V|)/\epsilon^2)$ and returns a valid solution with objective value no more than $(1+\epsilon)$
  times that of the optimum.
\end{thm}
\begin{proof}
  The proof is a simple application of Theorem 2.1 from \cite{fleischer} to the above Linear Program. A simple corollary of that theorem is that
  the algorithm finishes in $O(|V||E|\log(|V|)/\epsilon^2)$ steps\footnote{the constant $C$ in Theorem 2.1 of \cite{fleischer} is 1 in our
    case.} and in each step one has to find the most unsatisfied constraint among \eqref{lp:bilayer:1}-\eqref{lp:bilayer:3} given a current
  solution $\overline{\bf d}$. Each such step can be done by evaluating all the constraints in total time $|E|$.
\end{proof}

%%% Local Variables: 
%%% mode: latex
%%% TeX-master: "Regression"
%%% End: 

%\newpage
\comment{\section{The Push Algorithm}
\begin{algorithm}[ht]
\KwIn{Undirected tree $T=(V,E)$, with a vector of vertex weights
  $\mathbf{a} \in \mathbb{R}_+^n$}
\KwOut{A feasible vector of weights $\mathbf{x} \in \mathbb{R}_+^n$
  for $V$}
\BlankLine

\SetKwBlock{Begin}{begin}{end}

\label{alg1-app}
\SetKwFunction{KwPush}{Push-Path}
\SetKwFunction{KwImproveNode}{ImproveSubtree}
Let $v_1,v_2, \dots, v_{n-1}, v_n$ be the vertices in $T$ sorted in
post-order. \\

\For{$v \leftarrow 1$ to $n$}{
$x_v=max\{\sum_{u\in C(v_i)}x_u, a_v\}$\\%, \mathbf{y} \leftarrow \mathbf{x}$ \\
\KwImproveNode{$v$}
}
\BlankLine
\KwImproveNode{Vertex $u$} \\
\While{$\exists$ path $P$ from $u$ down to a leaf, and $\epsilon>0$ such that
  \KwPush($\mathbf{x}, P$,$\epsilon$)=$\epsilon$}{

\KwPush{$\mathbf{x}, p, \epsilon$} \\
%$\mathbf{y} \leftarrow \mathbf{x}$ \\ 
}

\BlankLine
\KwPush{Assignment $\mathbf{x}$, Path $P$, Non-negative real-value $\epsilon$} \\
\Begin{
Let $v_1,\ldots,v_k$ be the sequence of nodes on the $P$ from top to
bottom. \\
% \For{$i=1$ to $k$}{
% $x'_{v_i}=x_{v_i}$ 
% }

$old= \sum_{1 \leq i \leq k}|x_{v_i}-a_{v_i}|$ \\
$x_{v_1}=x_{v_1}-\epsilon$

\For{$i=2$ to $k$}{
$t=\sum_{u\in C(v_{i-1})}x_u-x_{v_{i-1}}$ \\
\If{$t > 0$}{
$x_{v_i}=x_{v_i}-t$
}
}
$new= \sum_{1 \leq i \leq k}|x_{v_i}-a_{v_i}|$ \\
\Return $old-new$
}
\caption{Push-Improve}
\end{algorithm} }
%\newpage

\section{Omitted figures and algorithms}\label{app:1}
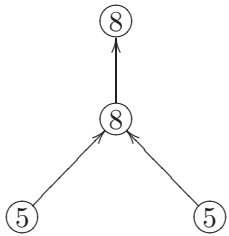
\begin{figure}[hp]
\begin{equation*}
\xymatrix{
 & *+[o][F-]{8} \\
& *+[o][F-]{8} \ar[u] \\
*+[o][F-]{5} \ar[ur] & & *+[o][F-]{5} \ar[ul]
}
\end{equation*}
\caption{A counter-example for the naive algorithm. Node annotations
  denote the $a$ values. The naive algorithm will obtain an objective
  function value of $4$, whereas the optimum value is $2$.}\label{fig1}
\end{figure}
%\section{The Depth-First-Search Algorithm for The $\ell_1$ Case On Trees}
\begin{algorithm}[hp]
\KwIn{Vertex $v$}
\BlankLine
\SetKwBlock{Begin}{begin}{end}
\SetKwFunction{KwSetParams}{Set-Params}
\label{alg2}
\SetKwFunction{KwPush}{Push-Search}
\Begin{
\tcc{If $x_v = a_v$ $T(v)$ is optimal}
\If{$x_v > a_v$}{
  \KwPush{$v, \infty, 0$} \\
}
}
\BlankLine
\KwPush{Vertex $u$, Non-negative real-value $\epsilon$, Non-negative
  Integer $\delta$} \\
\Begin{
\KwSetParams{$u, \epsilon, \delta$} \\
\If{$\epsilon_u = 0$}{
  \Return{$0$}
}
$sum \leftarrow 0$ \\ 
$\ell \leftarrow \min \{ \epsilon_u, x_u - \sum_{k:\text{child of
    $u$}}x_k \}$ \\
  \If{$\ell > 0$ and $\delta_u > 0$}{
     $x_u \leftarrow x_u - \ell $, $sum \leftarrow \ell$ \\
      \KwSetParams{$u, \delta, \epsilon-\ell$} 
    }
  \ForEach{$j \in c(u)$}{
    \If{$sum = \epsilon_u$}
    {
      \Return{0} \quad\quad \textbf{/* Speedup}
    }
    $t \leftarrow $\KwPush{$j, \epsilon_u,\delta_u$} \\
    $sum \leftarrow sum + t,\; x_u \leftarrow x_u - t$ \\
    \KwSetParams{$u,\delta,\epsilon - sum$} \\
   }
   \Return{$sum$}
}
\BlankLine
\KwSetParams{Vertex $i$, Non-negative integer $\delta$, Non-negative
  real-value $\epsilon$} \\
%\Begin{
  \If{$x_i > a_i$}{
    $\delta_i \leftarrow \delta + 1, \epsilon_i \leftarrow \min\{\epsilon, x_i - a_i\}$
  }
  \Else{
    $\delta_i \leftarrow \delta -1, \epsilon_i \leftarrow \min \{x_i, 
    \epsilon \}$ \\
  }
%}
\caption{The improved ImproveSubtree procedure}
\end{algorithm} 

%%% Local Variables: 
%%% mode: latex
%%% TeX-master: "Regression"
%%% End: 

% sigproc.bib is the name of the Bibliography in this case
% ACM needs 'a single self-contained file'!
%\balancecolumns
% \appendix
% \input{adinfluence-appendix.tex}

%Appendix A
%\section{Headings in Appendices}

%\balancecolumns
% That's all folks!
\end{document}